\documentclass[12pt]{article}

\usepackage{hansonstyle}
\usepackage{changepage}
\usepackage{cite}
\ytableausetup{centertableaux}
\usepackage[title]{appendix}

\begin{document}
\title{Investigations on Automorphism Groups of Quantum Stabilizer Codes}
\author{Hanson Hao}
\date{}
\maketitle

\begin{abstract}
The stabilizer formalism for quantum error-correcting codes has been, without doubt, the most successful at producing examples of quantum codes with strong error-correcting properties. In this paper, we discuss \emph{strong} automorphism groups of stabilizer codes, beginning with the analogous notion from the theory of classical codes. Two weakenings of this concept, the \emph{weak} automorphism group and \emph{Clifford-twisted} automorphism group, are also discussed, along with many examples highlighting the possible relationships between the types of ``automorphism groups". In particular, we construct an example of a $[[10,0,4]]$ stabilizer code showing how the Clifford-twisted automorphism groups might be connected to the Mathieu groups. Finally, nonexistence results are proved regarding stabilizer codes with highly transitive strong and weak automorphism groups, suggesting a potential inverse relationship between the error-correcting properties of a quantum code and the transitivity of those automorphism groups.
\end{abstract}

\newpage

\tableofcontents

\newpage

\section{Introduction}\label{sec1}

Although quantum computers are thought to be significantly more efficient than classical computers (e.g. Shor's algorithm for integer factorization), they are also inherently more susceptible to noise and breakdown processes. The construction of a \emph{quantum error-correcting code} (QECC) in \cite{Shor1995} demonstrated how protection against a single-qubit error was possible, which was a milestone towards the practical realization of quantum computers. However, the problem of finding ``good" QECCs is still difficult, and their theory is still in an early stage of development, compared to, say, the theory of classical error-correcting codes. We were interested in this very general question of finding ``good" QECCs, with a particular emphasis on investigating how ``symmetric" a QECC could be. We were also motivated by the recent paper of Harvey and Moore \cite{Harvey2020}, in which possible connections are found among QECCs, conformal field theories, and the Mathieu moonshine phenomena. Similar theories for classical codes revolve around finding their automorphism groups, which explains the motivations behind some of the definitions and results in Section \ref{sec2} and Section \ref{sec3}.

This paper is structured as follows: in the remainder of this section we will present a brief background of quantum error correction and \emph{stabilizer codes}, following \cite{Gottesman2009} and Chapter 10 of \cite{Nielsen2011}. We will put more emphasis on the mathematical structures and skip much of the physical intuition behind QECCs. In Section \ref{sec2} we will introduce three notions of automorphism groups of quantum codes, which we call \emph{strong}, \emph{weak}, and \emph{Clifford-twisted}. We prove some basic results regarding these different types of automorphism groups. The bulk of this section, however, will be devoted to giving examples of all three types of automorphism groups of small QECCs, along with explaining how they might be calculated by hand. Particularly interesting is the example in Section \ref{sec2.2.6}, which gives evidence of a possible connection between Clifford-twisted automorphism groups of stabilizer codes and the Mathieu groups. In Section \ref{sec3} we will show in certain cases that the strong and weak automorphism group of ``good" stabilizer codes cannot be highly transitive, which suggests that the Clifford-twisted automorphism groups might have the most interesting properties. Along the way, we mention some problems that we found interesting, but did not have time to think about.

\subsection*{Acknowledgements}

This research was supported by the Stanford Undergraduate Research Institute in Mathematics (SURIM) program during the summer of 2021. The author would like to thank Dr. Pawel Grzegrzolka for coordinating the program. The author also thanks his mentor, Dr. Daniel Bump, for introducing him to the area of quantum error correction, and also for many helpful conversations and insightful suggestions. His idea of investigating ``twisted automorphisms" (Definitions \ref{weakAutDef}, \ref{cliffordAut}) was particularly fruitful.

\subsection{Prelude: Classical Error Correction}\label{sec1.1}

We give a very brief overview of classical error correction to serve as motivation for the constructions of quantum error correction theory, which will be presented in Section \ref{sec1.2}. First, recall that in classical computing, the basic unit of computation is the \emph{bit}, which takes a state of either 0 or 1; i.e. an element of $\F_2$. Thus an $n$-bit state is just an element of $\F_2^n$.

The fundamental idea of classical error correction is repetition. In particular, if we want to transmit $k$ bits of information, we would repeat that information in such a way that it would be ``hard" to confuse slight moderations of the encoded information with an encoding of a different $k$ bits of information. For instance, if we wanted to send the bits 0 and 1 through a somewhat noisy channel, we could instead send the 3-bit sequences 000 and 111, so that even if there is a 1-bit error, we could use a ``majority rules" scheme to correct the error. As an explicit example, if we received 010, which is not a possible codeword, we could conclude with high probability that 000 was the intended message. In this manner, we are essentially embedding one bit into three by the identification of $\F_2^1$ with the 1-dimensional subspace $\{000,111\}$ in $\F_2^3$, so we call such a scheme a \emph{linear code}. We will soon see how the same ideas apply in the quantum world, although there are some caveats.

\subsection{Quantum Error Correction}\label{sec1.2}

\begin{definition}\label{qubit}
A \emph{qubit} is the basic unit of quantum computation. It is represented as the complex vector space $\C^2$ with basis vectors $\bf{0}$ and $\bf{1}$, so the state of a qubit is given by a nonzero linear combination (superposition) $a\mathbf{0}+b\mathbf{1}$.
\end{definition}

Some notational remarks are in order. First, we will write vectors/states in boldface, instead of the bra-ket convention popular in physics. Second, we will not bother with normalizing the state $a\mathbf{0}+b\mathbf{1}$ of a qubit so that $\abs{a}^2+\abs{b}^2=1$, as that does not affect our discussion.

\begin{definition}\label{nqubit}
A space of $n$ qubits is represented as the tensor product $(\C^2)^{\otimes n}=\underbrace{\C^2\otimes\ldots\otimes\C^2}_{n\text{ times}}$, where tensor products are taken over $\C$. A basis for this space is given by $$\{\mathbf{0\ldots 0}, \mathbf{0\ldots01},\ldots,\mathbf{1\ldots 1}\},$$ the $2^n$ binary vectors of length $n$.
\end{definition}

We will usually refer to such a space of $n$ qubits as $\C^{2^n}$.\\

The central idea of quantum error correction is to identify a $k$-qubit space with some $2^k$-dimensional linear subspace, called the \emph{codespace}, of an $n$-qubit space, called the \emph{ambient space}, where $n>k$. We will call states in the ambient $2^n$-dimensional space \emph{physical}, and states in the $2^k$-dimensional codespace \emph{logical}. For example, if we encode a one-qubit space into three qubits by the map $\mathbf{0}\mapsto\mathbf{000}, \mathbf{1}\mapsto\bf{111}$, we may refer to the physical state $\bf{000}$ as ``logical $\bf{0}$", denoted $\mathbf{0}_L$.

In this formalism, errors will just be linear operators acting on the codespace. To have error-correcting properties, this codespace should be redundant in some sense, but there is an added difficulty in that we are not allowed to create repetitions of quantum states, as we might do in classical computing. This is the \emph{no-cloning theorem}, as discussed in Theorem 1 of \cite{Gottesman2009}. Therefore, this codespace should consist of highly entangled states---``redundancy without repetition". Moreover, in general physical situations, one may encounter errors randomly affecting a single qubit, or a small number of qubits, of the ambient space. A good code allows such errors to be corrected.

%Therefore, this codespace will consist of highly entangled states (mathematically, we can get away with thinking of such states as having many nonzero terms when written in the basis $\{\mathbf{0\ldots 0}, \mathbf{0\ldots01},\ldots,\mathbf{1\ldots 1}\}$). REMOVED

We now state the quantum error-correction conditions as a ``black box". A full derivation can be found in Section 10.3 of \cite{Nielsen2011}.

\begin{theorem}\label{qecCondition}
Let $C\subset\C^{2^n}$ be a quantum code, and let $P$ be the orthogonal projection onto $C$. Then $C$ can correct a set of errors $\mathcal{E}=\{E_i\}$ if and only if there is a complex Hermitian matrix $(\alpha_{ij})$ such that
\begin{equation}\label{qecConditionEq}
PE_i^*E_jP=\alpha_{ij}P,
\end{equation}
where $E_i, E_j$ run over all operators in $\mathcal{E}$, and $^*$ is the conjugate transpose.
\end{theorem}

Using Theorem \ref{qecCondition}, one can show that if a code $C$ corrects a set of errors $\{E_i\}$, then it also corrects \emph{any linear combination} of the $E_i$. This is an extremely powerful observation, because instead of possibly having to correct a continuum of errors, it now suffices to focus only on correcting a discrete set of error operations that span the full set of errors we want to correct. Because of this observation, it makes sense to introduce the Pauli matrices:

\begin{definition}\label{pauliMatrix}
The four Pauli matrices are defined as follows:
\begin{equation}\label{paulis}
I=\begin{bmatrix}1&0\\0&1\end{bmatrix},\quad X=\begin{bmatrix}0&1\\1&0\end{bmatrix},\quad Z=\begin{bmatrix}1&0\\0&-1\end{bmatrix}, \quad Y=iXZ=\begin{bmatrix}0&-i\\i&0\end{bmatrix}. 
\end{equation}
\end{definition}

We call $X$ the \emph{bit flip} operator, since it sends $\bf{0}$ to $\bf{1}$ and vice versa. Similarly, we call $Z$ the \emph{phase flip} operator, since it sends $\bf{0}$ to itself, but $\bf{1}$ to $-\bf{1}$. Finally, $Y$ is a combined bit-phase operator, with an extra factor of $i$ thrown in to make the mathematics easier. The Pauli matrices act on the one-qubit space $\C^2$, but $n$-fold tensor products of them act on an $n$-qubit space in the natural way. For example, the tensor product $X\otimes Z\otimes I$ acts on $\C^{2^3}$ by sending $\bf{000}$ to $\bf{100}$, $\bf{111}$ to $-\bf{011}$, etc. In the rest of this paper, we will simply write tensor products of Pauli matrices as concatenations, so that the above $X\otimes Z\otimes I$ operator is simply $XZI$.

The Pauli matrices have remarkable properties that will be fully exploited in the rest of the paper. First, they are all both unitary and Hermitian, and they form a basis of the 2-by-2 complex matrices. So by the discussion after Theorem \ref{qecCondition}, if we wanted to, for instance, correct arbitrary errors occurring on the first qubit of a three-qubit system, we simply need to correct the errors $XII$, $ZII$, and $YII$ (along with $III$, which is automatic). Second, all of the Pauli matrices square to the identity, and they all commute or anticommute. In particular, two Pauli matrices anticommute if and only if they are different nonidentity matrices. This implies that the Pauli matrices form a projective representation of the four-group $\Z_2\times\Z_2$, which is why they are chosen as our basis of the 2-by-2 complex matrices. Since complex scalars, in general, do not concern us, the fact that the Pauli matrices act like the elements of $\Z_2\times\Z_2$ is of great utility.

\subsection{The Pauli Group and Stabilizer Codes}\label{sec1.3}

Now, although the quantum error-correction condition \ref{qecConditionEq} is easy to verify for any particular code and set of errors, it is difficult to actually construct a code correcting a given set of error operations, particularly if that set is large. Indeed, we are usually interested in correcting sets of errors such as ``all one-qubit errors," which necessitates correcting an error set of size $3n$ if the ambient space is an $n$-qubit space (we need to correct an $X$, $Z$, and $Y$ error at each physical qubit). The stabilizer formalism introduced by Gottesman \cite{Gottesman1997} provides a convenient workaround to this problem.

\begin{definition}\label{pauliGroup}
The one-qubit \emph{Pauli group}, denoted $G_1$, is the group of matrices generated by $X$, $Y$, and $Z$. This is a group of order 16: $$G_1=\{\pm I, \pm iI, \pm X, \pm iX, \pm Y, \pm iY, \pm Z, \pm iZ\}.$$
\end{definition}

Note that the elements of $G_1$ are just Pauli matrices multiplied by some phase factor, which is a fourth root of unity. All elements in $G_1$ square to plus or minus the identity, and all elements commute or anticommute. This means that $G_1$ is ``almost" an elementary abelian 2-group: its commutator subgroup is $\gen{-I}$, and $G_1/[G_1,G_1]$ is indeed an elementary abelian 2-group of order $2^{2n+1}$. It will also be useful to speak of $G_1$ mod its center $Z(G_1)=\gen{iI}$, so that we can consider elements without worrying about scalar multiples. We denote $G_1/Z(G_1)$ by $P_1$.

It is natural to generalize the Pauli group to $n$-fold tensor products:

\begin{definition}\label{pauliNGroup}
The $n$-qubit Pauli group, denoted $G_n$, is the group of matrices generated by $n$-fold tensor products of the Pauli matrices. This is a group of order $4^{n+1}$, since there are $4^n$ possible tensor products, along with 4 possible phases (the fourth roots of unity).
\end{definition}

Again, $[G_n,G_n]=\gen{-I}$, $G_n/[G_n,G_n]$ is an elementary abelian 2-group, and $Z(G_n)=\gen{iI}$. We denote $G_n/Z(G_n)$ by $P_n$.

Before moving on to stabilizer codes, we mention an extremely useful way of writing elements of $P_n$, known as the \emph{check matrix}. We see that an element $s\in P_n$ can uniquely be written as $$X^{a_1}Z^{b_1}\otimes X^{a_2}Z^{b_2}\otimes\ldots\otimes X^{a_n}Z^{b_n},$$ where the $a_i, b_i$ are 0 or 1. Therefore we represent $s$ as
\begin{equation}\label{checkMatrix}
(a|b)=
\left[
\begin{array}{cccc|cccc}
a_1&a_2&\ldots&a_n&b_1&b_2&\ldots&b_n
\end{array}
\right],
\end{equation}
which can be thought of as a vector in $\F_2^{2n}$. For example, the element $IXYZ$ is represented as
\[
\left[
\begin{array}{cccc|cccc}
0&1&1&0&0&0&1&1 \\
\end{array}
\right].
\]

The utility of this representation becomes clear when we define a symplectic bilinear form over $\F_2^{2n}$ given by
\begin{equation}\label{sympForm}
\gen{(a|b),(c|d)}=a\cdot d+b\cdot c,
\end{equation}
where $\cdot$ is the usual dot product in $\F_2^n$. Then it is straightforward to check that if $s, s'\in P_n$ have check matrix representations $(a|b)$, $(a'|b')$ respectively, $s$ and $s'$ commute if and only if $\gen{(a|b),(a'|b')}=0$. Since $P_n$ is simply $G_n$ modded out by its center, the same is true for any lifts of $s, s'$ to elements in $G_n$. The check matrix representation gives us a compact way of writing set of elements in $G_n$ or $P_n$, which will be useful later on.

We are now ready to define stabilizer codes. The key is to consider the action of $G_n$ on $\C^{2^n}$, and consider the fixed points of a subgroup of $G_n$.

\begin{definition}\label{stabilizerCodes}
Let $S$ be an \emph{abelian} subgroup of $G_n$ \emph{not containing} $-I$. Then
\begin{equation}\label{stabilizerCodeDef}
C(S)=\{v\in\C^{2^n}:sv=v\text{ for all }s\in S\}
\end{equation}
is the \emph{stabilizer code} corresponding to $S$. We call $S$ the \emph{stabilizing subgroup} corresponding to $C(S)$.
\end{definition}
It is easy to check that $C(S)$ is a subspace of the ambient space.

It is important to make a few remarks regarding this definition. First, $C(S)$ must contain \emph{all} vectors in $\C^{2^n}$ that are fixed by everything in $S$. In particular, for a code $C'$ to be called a stabilizer code, there must be some abelian subgroup $S$ of $G_n$, not containing $-I$, such that $C'=C(S)$. It is not enough for $C'\subset C(S)$. Second, we do not want $S$ to contain $-I$. Otherwise, if $-I\in S$, then for any $v\in C(S)$, we have $v=-Iv\Ra v=0$, so that $C(S)$ is trivial. Similarly, we require $S$ to be abelian: if $s,s'\in S$ do not commute, then they must anticommute, so that for any $v\in C(S)$, we have $v=(ss')v=(-s's)v=-(s'sv)=-v\Ra v=0$. Note that because elements in $G_n$ either square to $I$ or $-I$, these two conditions imply that $S$ is an elementary abelian 2-group. Also, elements in $S$ may only have a scalar factor of plus or minus 1, lest they square to $-I$.

Given the construction of stabilizer codes, we would like to somehow connect a stabilizing subgroup $S\subset G_n$ with the dimension of $C(S)$, as well as with the errors that $C(S)$ can correct. Let us first answer the former question. First, note that $S$ can have size at most $2^n$. To see why, suppose $S$ has $m$ independent generators (so $\abs{S}=2^m$), which we write in the check matrix format. We may think of $S$ as an $m$-dimensional subspace of the $\F_2$-vector space $\F_2^{2n}$. Because these generators all commute, $S$ is self-orthogonal with respect to the bilinear form in \ref{sympForm}; that is, $S\subseteq S^\perp$. But $\dim S+\dim S^\perp=2n$, so that $m=\dim S\leq n$.

With $m\leq n$ in mind, it now makes sense to state the following proposition:
\begin{proposition}\label{stabCodeDim}
Let $C$ be the stabilizer code corresponding to a stabilizing subgroup $S\subset G_n$. If $S$ has $m$ independent generators, or equivalently $\abs{S}=2^m$, then $C$ has dimension $2^{n-m}$; that is, it encodes $k\coloneqq n-m$ logical qubits.
\end{proposition}
\begin{proof}
We claim that the orthogonal projection onto $C$ is given by $P=\frac{1}{\abs{S}}\sum_{s\in S}s$. We must check that
\begin{enumerate}
\item $P=P^*$,
\item $Pv=v$ for all $v\in C$,
\item $P^2=P$.
\end{enumerate}
For (1), we notice that the elements of $S$ are, up to sign, tensor products of the Pauli matrices, which are all Hermitian. Hence any $s\in S$ is Hermitian, so $P$ is as well. For (2), we calculate $Pv=\frac{1}{\abs{S}}\sum_{s\in S}sv=\frac{1}{\abs{S}}\sum_{s\in S}v=v$, since $v\in C$ implies $sv=v$ for all $s\in S$. For (3), we have $$P^2=\left(\frac{1}{\abs{S}}\sum_{s\in S}s\right)\left(\frac{1}{\abs{S}}\sum_{t\in S}t\right)=\frac{1}{\abs{S}}\sum_{s\in S}\left(\frac{1}{\abs{S}}\sum_{t\in S}st\right)=\frac{1}{\abs{S}}\sum_{s\in S}P=P,$$ since summing over all $st\in S$, for $s$ fixed, is the same as summing over all $t\in S$.

So because $P$ is an orthogonal projection onto $C$, we have $\dim C=\text{rk }P=\tr P$. Notice that $X$, $Y$, and $Z$ are traceless. Using this and the fact that $-I\not\in S$, we see that the only term in $\sum_{s\in S}s$ that has nonzero trace is the identity. Hence $\tr P=\frac{1}{\abs{S}}(\tr I)=\frac{2^n}{\abs{S}}=2^{n-m}$.
\end{proof}
In this situation, we call $C$ an $[[n,k]]$ code, where the double brackets are meant to distinguish the quantum code $C$ from classical codes. From now on, $S$ will always be assumed to an abelian subgroup of $G_n$ not containing $-I$, and we will just write $C$ for $C(S)$ if there is no ambiguity in the stabilizing subgroup $S$.

We now discuss the error-correcting properties of $C$. First, we need to define a slight abuse of notation:
\begin{definition}\label{NS-S}
Let $N(S)$ be the normalizer of $S$ in $G_n$. We say that an element $p\in P_n$ is in $N(S)-S$ if there is some lift $g\in G_n$ of $p$ in $N(S)-S$.
\end{definition}
The idea behind this definition is that we can and should disregard the scalar phase of any element of $G_n$, since we know that if $C$ corrects some error $E$, it corrects any scalar multiple of $E$. Notice that $N(S)$ is equal to the centralizer of $S$, since two elements of $G_n$ either commute or anticommute, and $-I\not\in S$.

Along with this definition, we must reinterpret the quantum error-correction conditions, Theorem \ref{qecCondition}, in the language of stabilizer codes.
\begin{theorem}\label{qecStabilizer}
Let $C$ be the stabilizer code corresponding to a subgroup $S$. If $\{E_i\}$ is a set of error operators in $G_n$ such that $E_i^*E_j\not\in N(S)-S$ for all $i$ and $j$, then the $\{E_i\}$ are correctable.
\end{theorem}
\begin{proof}
See Theorem 10.8, \cite{Nielsen2011}.
\end{proof}

The criterion in Theorem \ref{qecStabilizer} is much easier to use than that in Theorem \ref{qecCondition}. Instead of having to construct the projection matrix onto our code and do various matrix calculations, the error-correcting properties of a stabilizer code can be computed only using knowledge of $G_n$. Because of this, we introduce a few more definitions regarding elements of $G_n$ and $P_n$:
\begin{definition}\label{wt}
The \emph{weight} of an element $p\in P_n$ is the number of non-identity tensor factors in $p$.
\end{definition}
For instance, the element $XZIIZ\in P_5$ has weight 3, since it has three non-identity tensor factors.
\begin{definition}\label{dist}
Let $C$ be an $[[n,k]]$ stabilizer code corresponding to a subgroup $S\subset G_n$, where $k\geq1$. Then the \emph{distance} $d$ of $C$ is defined as
\begin{equation}
d=\min\{\weight(p):p\in P_n, p\in N(S)-S\}.
\end{equation}
Recall that we write $p\in N(S)-S$ as in the sense of Definition \ref{NS-S}. In this case, we call $C$ a $[[n,k,d]]$ code.
\end{definition}
\begin{remark}\label{distk=0}
In the degenerate case $k=0$, where $C$ is 1-dimensional (i.e. a single state), $N(S)-S$ is empty, so the above definition does not make sense. We follow the convention of \cite{Calderbank1997} and say that the distance is
\begin{equation}
d=\min\{\weight(p):p\in S, p\not=I\}.
\end{equation}
%In general, we will try to avoid the $k=0$ case in our discussions.
\end{remark}

From these definitions, we see that the product of any two weight $w$ operators has weight at most $2w$, and that the conjugate transpose does not change the weight of an operator. Then it follows from Theorem \ref{qecStabilizer} that any code with distance greater than $2w$ can correct errors affecting any $w$ qubits.
\begin{example}\label{513Code1}
Consider a subgroup $S\subset G_5$ generated by $$\{XZZXI, IXZZX, XIXZZ, ZXIXZ\}.$$ These elements all pairwise commute, and $S$ does not contain $-I$, so the stabilizer code $C$ corresponding to $S$ encodes $5-4=1$ logical qubit. Moreover, one can directly verify that the distance of $S$ is 3---for instance, $XIZIX\in N(S)-S$, but there are no elements $p\in P_n$ of lesser weight in $N(S)-S$. This shows that $C$ can correct any error affecting one qubit, since $3>2\cdot1$.

Using the projection matrix of a stabilizer code as in Proposition \ref{stabCodeDim}, we can find a basis for $C$:
\begin{align*}
\mathbf{0}_L&=\mathbf{00000}+\mathbf{10010}+\mathbf{01001}+\mathbf{10100}\\
&+\mathbf{01010}-\mathbf{11011}-\mathbf{00110}-\mathbf{11000}\\
&-\mathbf{11101}-\mathbf{00011}-\mathbf{11110}-\mathbf{01111}\\
&-\mathbf{10001}-\mathbf{01100}-\mathbf{10111}+\mathbf{00101}\\
\mathbf{1}_L&=(XXXXX)\mathbf{0}_L\\
&=\mathbf{11111}+\mathbf{01101}+\mathbf{10110}+\mathbf{01011}\\
&+\mathbf{10101}-\mathbf{00100}-\mathbf{11001}-\mathbf{00111}\\
&-\mathbf{00010}-\mathbf{11100}-\mathbf{00001}-\mathbf{10000}\\
&-\mathbf{01110}-\mathbf{10011}-\mathbf{01000}+\mathbf{11010}.
\end{align*}
\end{example}

\begin{remark}\label{smallestRmk}
The aforementioned stabilizer code is the smallest able to correct one error on any qubit, in terms of the number $n$ of physical qubits.
\end{remark}

More examples of small stabilizer codes, as well as more details about their theory, can be found in Section 10.5 of \cite{Nielsen2011}. Many more examples of stabilizer codes, including codes with larger parameters $[[n,k,d]]$, can be found at \cite{codetables}. We will reference those tables of quantum codes in the below sections.

\section{Automorphism Groups of Codes}\label{sec2}

\subsection{Definitions and Basic Results}\label{sec2.1}

Just like in the theory of classical codes, we can define the automorphism group of a QECC:

\begin{definition}\label{strongAutDef}
Let $C$ be a quantum code. There is a natural action of $S_n$ on the $n$-qubit ambient space $\C^{2^n}$. We define the \emph{strong automorphism group} of $C$ as
\begin{equation}
\Autstr(C)=\{\sigma\in S_n:\sigma(C)=C\}.
\end{equation}
\end{definition}
We are particularly interested in the case where $C$ is a stabilizer code with stabilizing subgroup $S$. Indeed, in this case we may also consider the action of $S_n$ on $G_n$ given by permuting tensor factors. This action is equivalently given by $\sigma(g)=M_\sigma gM_\sigma^{-1}$, where $g\in G_n$ and $M_\sigma^{-1}$ is the permutation matrix associated to $\sigma$ in the natural basis of $\C^{2^n}$. Then the above condition for $\Autstr(C)$ can be rephrased in terms of $S$:

\begin{proposition}\label{strongAutProp}
Let $C$ be a nontrivial (i.e. nonzero) stabilizer code corresponding to a subgroup $S\subset G_n$. Then $\sigma\in\Autstr(C)$ if and only if $\sigma(S)=S$.
\end{proposition}
\begin{proof}
Notice that $\sigma(S)$ fixes $\sigma(C)$ pointwise. Hence if $\sigma(S)=S$, it follows that $\sigma(C)\subseteq C$. But $\sigma$ acts as an invertible linear map $C\to\sigma(C)$, which implies $\sigma(C)=C$. This proves the ``if" direction. Conversely, if $\sigma(C)=C$, then $\sigma(C)$ is fixed pointwise by $S$. It is also fixed pointwise by $\sigma(S)$, so the same is true for the join $\gen{S,\sigma(S)}$ of the two subgroups. Because $C$ is nontrivial, $\gen{S,\sigma(S)}$ must be abelian and must not contain $-I$. Then by Proposition \ref{stabCodeDim}, $\gen{S,\sigma(S)}$ must have the same size as $S$, which implies $\sigma(S)\subseteq S$. But $\sigma$ is bijective, so we have equality, proving the reverse direction.
\end{proof}

Since all $\sigma\in S_n$ act as invertible linear maps on $S$, we can simplify the result of Proposition \ref{strongAutProp} to a form that is more suitable for actually computing strong automorphism groups.

\begin{corollary}\label{strongAutcor}
Let $C$ be a nontrivial stabilizer code corresponding to a subgroup $S\subset G_n$. Let $S$ be generated by $g_1,\ldots, g_m$. Then $\sigma\in\Autstr(C)$ if and only if $\sigma(g_i)\in S$ for each $i$.
\end{corollary}

\begin{example}\label{311code}
Let $S\subset G_3$ be generated by $XZZ$ and $ZXZ$, so that $S=\{III, XZZ, ZXZ, YYI\}$. Using the Kronecker product for matrices, we can calculate the projection onto $C=C(S)$ as
\[
P=\frac{1}{\abs{S}}\sum_{s\in S}s=\frac{1}{4}\begin{bmatrix}
1&0&1&0&1&0&-1&0\\
0&1&0&-1&0&-1&0&-1\\
1&0&1&0&1&0&-1&0\\
0&-1&0&1&0&1&0&1\\
1&0&1&0&1&0&-1&0\\
0&-1&0&1&0&1&0&1\\
-1&0&-1&0&-1&0&1&0\\
0&-1&0&1&0&1&0&1
\end{bmatrix},
\]
so that $C$ is spanned by $\mathbf{0}_L=\frac{1}{2}\left(\mathbf{000}+\mathbf{010}+\mathbf{100}-\mathbf{110}\right)$ and $\mathbf{1}_L=\frac{1}{2}\left(\mathbf{001}-\mathbf{011}-\mathbf{101}-\mathbf{111}\right)$. Checking each of the permutations in $S_3$, we obtain $\Autstr(C)=\{(1), (12)\}\cong\Z_2$. Indeed, these are also the only permutations $\sigma$ that satisfy $\sigma(S)=S$, as is easily seen.
\end{example}

\begin{remark}\label{2.1rmk1}
At this point, it is worth mentioning that the strong automorphism group of a stabilizer code cannot be characterized by the parameters $[[n,k,d]]$. For instance, let $S=\{IIII, XXXX, ZZZZ, YYYY\}$ and $S'=\{IIII, XXZZ, YYXX, ZZYY\}$ be subgroups of $G_4$. Then the corresponding stabilizer codes $C$ and $C'$ both have parameters $[[4,2,2]]$. But $\Autstr(C)=S_4$, while $\Autstr(C')=\{(1),(12),(34),(12)(34)\}\cong\Z_2\times\Z_2$.
\end{remark}

It turns out that the notion of ``strong automorphism group" does not produce many interesting examples for stabilizer codes with small $n$; in particular, it seems that for most codes, the strong automorphism group is usually small compared to the full symmetric group. We will see more examples in Section \ref{sec2.2}, and we will give some reasons as to why this general heuristic might be true in Section \ref{sec3}. Therefore we now introduce an extension of the automorphism group concept.

\begin{definition}\label{weakAutDef}
Let $C$ be a quantum code. We define the \emph{weak automorphism group} of $C$ as
\begin{equation}
\Autweak(C)=\{\sigma\in S_n:\sigma(C)=\gamma_\sigma\cdot C \text{ for some }\gamma_\sigma\in G_n\}.
\end{equation}
\end{definition}
That is, $\sigma(C)$ is ``almost" $C$, where we allow a twist of the vectors in $C$ by some element $\gamma_\sigma$ of the Pauli group (which is allowed to vary depending on $\sigma$, and is not necessarily unique). We should check that $\Autweak(C)$ is actually a group. Indeed this is the case, since we can view $\Autweak(C)$ as a semidirect product of sorts: if $\sigma,\tau\in\Autweak(C)$ correspond to $\gamma_\sigma,\gamma_\tau\in G_n$, then
\begin{equation}
\sigma\tau(C)=(\sigma(\gamma_\tau)\cdot\gamma_\sigma)\cdot C,
\end{equation}
as can be easily verified. Recall that the action of $\sigma$ on $G_n$ by permuting tensor factors is the same as the matrix action by conjugation, $g\mapsto M_{\sigma}gM_{\sigma}^{-1}$. Hence $\sigma\tau\in G_n$, since we can set $\gamma_{\sigma\tau}=\sigma(\gamma_\tau)\cdot\gamma_\sigma$.

We now want to somewhat justify the labels ``strong automorphism" and ``weak automorphism". In Proposition \ref{strongAutProp}, we were able to give a characterization of strong automorphisms in terms of the stabilizing subgroup of a stabilizer code. We can do the same for weak automorphisms:
\begin{proposition}\label{weakAutProp}
Let $C$ be a nontrivial stabilizer code corresponding to a subgroup $S\subset G_n$. Then $\sigma\in\Autweak(C)$ if and only if $\sigma(S)\subset S\cup -S$, where $-S=\{-s:s\in S\}$.
\end{proposition}

\begin{corollary}\label{weakAutcor}
Since $\sigma$ is bijective, the following conditions are equivalent to the one in Proposition \ref{weakAutProp}:
\begin{enumerate}
\item If $S$ is generated by $g_1,\ldots, g_m$, then $\sigma\in\Autweak(C)$ if and only if $\sigma(g_i)$ or $-\sigma(g_i)$ is in $S$ for each $i$.
\item If $\pi$ is the projection $G_n\to G_n/\gen{-I}$, then $\pi(\sigma(S))=\pi(S)$; that is, $\sigma(S)$ and $S$ are the same subgroup up to sign.
\end{enumerate}
\end{corollary}

\begin{proof}[Proof of Proposition \ref{weakAutProp}]
First, if $\sigma\in\Autweak(C)$, then $\sigma(C)=\gamma_\sigma\cdot C$ for some $\gamma_\sigma\in G_n$. Because $C$ and $\sigma(C)$ have the same dimension, it follows from dimension considerations (Proposition \ref{stabCodeDim}) that $\sigma(S)$ is exactly the set of elements stabilizing $\sigma(C)$. Similarly, because $\gamma_\sigma$ is invertible, it follows that $\gamma_\sigma S\gamma_\sigma^{-1}$ is exactly the set of elements stabilizing $\sigma(C)$, so $\sigma(S)=\gamma_\sigma S\gamma_\sigma^{-1}$. But $\gamma_\sigma$ commutes or anti-commutes with everything in $S$, so $\sigma(S)=\gamma_\sigma S\gamma_\sigma^{-1}\subset S\cup-S$.

Conversely, suppose $\sigma(S)\subset S\cup-S$. $\sigma$ is an isomorphism of abelian groups, and $\sigma(S)$ does not contain $-I$ (otherwise $\sigma^{-1}(-I)=-I\in S$, contradiction). Since $S\cup-S=\{\pm I,\pm a_1,\pm a_2,\ldots\}$ for an enumeration of the elements of $S$, by counting and the fact that $-I\not\in\sigma(S)$, we see that plus or minus $s$ is in $\sigma(S)$ for any $s\in S$. In particular, if $s_1,\ldots, s_m$ are an independent set of generators for $S$, we may find unique $r_1,\ldots, r_m\in S$ such that $\sigma(r_i)$ equals plus or minus $s_i$. It follows that the $\{\sigma(r_i)\}$ generate $\sigma(S)$. % (basically a linear algebra argument, since $\sigma(S)\cong S$ is a vector space over $\F_2$)

Set $e_i=\sigma(r_i)s_i$, so $e_i\in\{\pm1\}$. Renumber the $r$'s, $s$'s, and $e$'s so that $e_1=\ldots=e_a=1$, and $e_{a+1}=\ldots=e_m=-1$. We would like to produce an element $\gamma\in G_n$ such that $\gamma$ commutes with $s_1,\ldots, s_a$, and anticommutes with $s_{a+1},\ldots, s_m$ (if $a=m$, then take $\gamma=I$, so we can assume $a<m$). It suffices to find a $\gamma$ that commutes with $\{s_1,\ldots, s_a, s_{a+1}s_{a+2}, s_{a+1}s_{a+3},\ldots, s_{a+1}s_m\}$, and anticommutes with $s_{a+1}$.

Recall the check matrix representation of elements in $G_n$ (where we disregard elements of the center $\gen{iI}$), so we may consider $S$ as an $m$-dimensional subspace of $\F_2^{2n}$. It follows that the elements in $S$ commuting with all of $\{s_1,\ldots, s_a, s_{a+1}s_{a+2}, s_{a+1}s_{a+3},\ldots, s_{a+1}s_m\}$ form a subspace $V$ of dimension $2n-(m-1)$, which is the orthogonal complement with respect to the symplectic bilinear form \ref{sympForm}. Similarly, the elements in $S$ commuting with all of $\{s_1,\ldots, s_a, \bld{s_{a+1}}, s_{a+1}s_{a+2}, s_{a+1}s_{a+3},\ldots, s_{a+1}s_m\}$ form a subspace $W$ of dimension $2n-m$. Therefore we may find some element in $V$ but not in $W$, and the corresponding $\gamma\in G_n$ (taking the phase scalar factor to be 1) is the desired element.

It follows that $\gamma S\gamma^{-1}\cong S$ is generated by $$\{s_1,\ldots, s_a,-s_{a+1},\ldots, -s_m\}=\{\sigma(r_1),\ldots,\sigma(r_m)\},$$ so that $\gamma S\gamma^{-1}=\sigma(S)$. $\gamma S\gamma^{-1}$ is the stabilizer for $\gamma\cdot C$, and $\sigma(S)$ is the stabilizer for $\sigma(C)$, so $\gamma\cdot C=\sigma(C)$. Setting $\gamma_\sigma\coloneqq\gamma$, we have $\sigma\in\Autweak(C)$.
\end{proof}

This somewhat justifies the ``strong" and ``weak" labels, since Propositions \ref{strongAutProp} and \ref{weakAutProp} show that a strong automorphism of $C(S)$ sends $S$ to itself, while a weak automorphism sends $S$ to itself but ``disregarding signs". Of course the strong automorphism group is a subgroup of the weak automorphism group.

\begin{example}\label{301code}
Let $S\subset G_3$ be generated by $XXX$, $YYI$, and $ZXZ$, so that \\$S=\{III, XXX, YYI, ZXZ, -ZZX, -YIY, XZZ, -IYY\}$. Then Corollaries \ref{strongAutcor} and \ref{weakAutcor} show that the strong automorphism group of $C=C(S)$ is $\{(1), (12)\}\cong\Z_2$, while the weak automorphism group is all of $S_3$.
\end{example}

This example shows that this weaker notion of automorphism can enlarge the strong automorphism group, and we will see in some later examples that $\Autweak(C)$ can be quite a bit larger than $\Autstr(C)$ for certain stabilizer codes $C$. This example also shows that $\Autstr(C)$ is not necessarily normal in $\Autweak(C)$. So we may pose the following (somewhat vague) question:

\begin{question}\label{q1}
Can we describe the relationship between $\Autstr(C)$ and $\Autweak(C)$ for a given stabilizer code $C=C(S)$?
\end{question}

We extend the notion of an automorphism group one last time. For this, we need to introduce the (one-qubit) \emph{Clifford group} $L_1$, following the terminology of \cite{Bolt1961} and \cite{Bolt1961-2}.

\begin{definition}\label{cliffordGroup}
The (one-qubit) \emph{Clifford group} $L_1$ is the subgroup of the normalizer of $G_1$ in $U(2)$ that contains entries from $\Q\left(\frac{1+i}{\sqrt{2}}\right)=\Q(e^{\frac{\pi i}{4}})$. It is generated by $G_1$, $H=\frac{1}{\sqrt{2}}\begin{bmatrix}1&1\\1&-1\end{bmatrix}$, and $S=\begin{bmatrix}1&0\\0&i\end{bmatrix}$.
\end{definition}

In the terminology of quantum logic gates, $H$ is the \emph{Hadamard gate}, and $S$ is the \emph{phase shift gate} where the shift is by $\frac{\pi}{4}$. Note that the conjugation actions of $H$ and $S$ on $X$, $Y$, and $Z$ are as follows:
\begin{center}
\begin{tabular}{|c|c|}
\hline
$HXH^{-1}=Z$&$SXS^{-1}=Y$\\
\hline
$HZH^{-1}=X$&$SZS^{-1}=Z$\\
\hline
$HYH^{-1}=-Y$&$SYS^{-1}=-X$\\
\hline
\end{tabular}.
\end{center}

We can extend the action of $L_1$ on $G_1$ to $n$ qubits as follows, which gives us our second weakening of the strong automorphism group:

\begin{definition}\label{nCliffordGroup}
We define the $n$-qubit \emph{diagonal Clifford group} $L_n$ as 
\begin{equation}
L_n=\{A_1\otimes\ldots\otimes A_n: A_i\in L_1\}.
\end{equation}
\end{definition}

$L_n$ has a natural conjugation action on $G_n$, which motivates the following definition.

\begin{definition}\label{cliffordAut}
Let $C$ be a stabilizer code corresponding to the subgroup $S\subset G_n$. We define the \emph{Clifford-twisted automorphism group} of $C$ as
\begin{equation}
\Autclif(C)=\{\sigma\in S_n:\sigma(S)=\lambda_\sigma S\lambda_\sigma^{-1}\text{ for some }\lambda_\sigma\in L_n\}.
\end{equation}
\end{definition}
Since $L_n$ contains $G_n$, we have the chain of inclusions $\Autclif(C)\supseteq\Autweak(C)\supseteq\Autstr(C)$ for a stabilizer code $C$. We are somewhat justified in calling $\Autclif(C)$ an automorphism group: it is a group for the same reason that $\Autweak(C)$ is a group, and because the elements of $L_n$ act as automorphisms of $G_n$, it is not hard to show that if $\sigma\in\Autclif(C)$, then $C$ and $\sigma(C)$ have the same $[[n,k,d]]$ parameters.

To aid computation, we make some useful observations involving Definition \ref{cliffordAut}. First, if $g_1,\ldots,g_m$ generate $S$, then for $\sigma$ to be in $\Autclif(C)$ it suffices to verify that $\sigma(g_i)\in\lambda_\sigma S\lambda_\sigma^{-1}$ for each $g_i$. This can be further weakened: it suffices that plus or minus $\sigma(g_i)$ is in $\lambda_\sigma S\lambda_\sigma^{-1}$, since if we are only off by a sign, we can simultaneously correct for that by conjugating by an appropriate element of $G_n$, as in the proof of Proposition \ref{weakAutProp}. Finally, because we can disregard signs, we can interpret the action of $L_1$ on $G_1$ (or each tensor factor of $G_n$) as $S_3$ acting on the non-identity Pauli matrices. For instance, the above table shows that $H$ induces the permutation $(XZ)$ and $S$ induces the permutation $(XY)$ on the set $\{X,Y,Z\}$, and various combinations of $H$ and $S$ induce the other permutations of $S_3$. So we can rephrase Definition \ref{cliffordAut} in a more verbose, but more intuitive, manner:

\begin{corollary}\label{cliffordAut2}
Let $C$ be a stabilizer code corresponding to the subgroup $S\subset G_n$, and consider the componentwise action of $(S_3)^n\coloneqq\underbrace{S_3\times\ldots\times S_3}_{n\text{ times}}$ on elements of $G_n$ (i.e. on each tensor factor, $S_3$ acts on the three non-identity Pauli matrices). Then $\sigma\in\Autclif(C)$ if and only if there is some $\rho_\sigma\in(S_3)^n$ such that for each generator $g_i$ of $S$, $\rho_\sigma\cdot(\sigma(g_i))\in S\cup -S$.
\end{corollary}
Note that we use $\cdot$ to emphasize the fact that $\sigma$ and $\rho_\sigma$ have very different actions: $\sigma$ permutes the tensor factors in each element, while $\rho_\sigma$ permutes $X$, $Y$, and $Z$ in each tensor factor.

Let us use Corollary \ref{cliffordAut2} in an example. Recall from Remark \ref{2.1rmk1} that if $C$ is the stabilizer code corresponding to the subgroup $S\subset G_4$ generated by $XXZZ$ and $YYXX$, then $\Autstr(C)=\{(1),(12),(34),(12)(34)\}\cong\Z_2\times\Z_2$. Using Corollary \ref{weakAutcor}, it is easy to see that $\Autweak(C)$ is the same group. However, we now show that $\Autclif(C)$ is the full $S_4$.
\begin{example}\label{422Clif}
We want to show that $\Autclif(C)=S_4$, and we already know $(12)\in\Autclif(C)$, so it suffices to show that $(1234)\in\Autclif(C)$. Applying this permutation to the generators $XXZZ$ and $YYXX$ gives $ZXXZ$ and $XYYX$. The latter two elements can be obtained from the former two by applying the permutation $(XZY)$ in the first tensor factor, and $(XYZ)$ in the third. Hence $(1234)\in\Autclif(C)$.
\end{example}

\subsection{Examples of Automorphism Groups}\label{sec2.2}

In this section we give many examples of $\Autstr(C), \Autweak(C)$, and $\Autclif(C)$ for various small stabilizer codes $C$ with large distance. Most codes will be taken from \cite{codetables}. The majority of these examples were calculated by hand and then verified using a SAGE program (see Appendix \ref{SAGEcode}). These hand calculations are somewhat tedious \textit{ad hoc} processes, so we will only try to give a general outline of how they were done. We will also make some observations about the more noteworthy automorphism groups. However, we should again remark that the automorphism groups of a stabilizer code cannot be completely determined by the parameters $[[n,k,d]]$, so these examples should be taken as interesting specific cases rather than showcasing a general phenomenon.

\subsubsection{A $[[5,1,3]]$ code}\label{sec2.2.1}

Recall from Example \ref{513Code1} that the subgroup $S\subset G_5$ generated by
\begin{align*}
&XZZXI\\
&IXZZX\\
&XIXZZ\\
&ZXIXZ
\end{align*}
gives a $[[5,1,3]]$ stabilizer code $C$, and it is the smallest able to correct one error on any qubit in terms of the number $n$ of physical qubits. We find that $\Autstr(C)=\Autweak(C)=\gen{(12345),(25)(34)}\cong D_{10}$, the dihedral group of order 10. One could compute these groups by writing out the elements of $S$ in full:
\begin{align*}
&XZZXI\quad XYIYX\quad ZIZYY\quad ZYYZI\\
&IXZZX\quad IZYYZ\quad YXXYI\quad YIYXX\\
&XIXZZ\quad YYZIZ\quad IYXXY\quad ZZXIX\\
&ZXIXZ\quad XXYIY\quad YZIZY\quad IIIII
\end{align*}
Then it is clear that $\Autstr(C)=\Autweak(C)$, since all of the elements have the same phase $+1$, and moreover any strong automorphism must take the four generators of $S$ to a cyclic permutation of $XZZXI$. It should follow pretty easily that $\Autstr(C)=\gen{(12345),(25)(34)}$. We call stabilizer codes $C$ with $(12\ldots n)\in\Autstr(C)$ (or more generally, any $n$-cycle in place of $(12\ldots n)$) \emph{cyclic}, and they are of particular interest (see \cite{Gottesman1997}).

It is much harder to come up with a systematic approach for calculating the Clifford-twisted automorphism group, so we must resort to ``guessing" elements in $\Autclif(C)$. It is a good idea to begin with testing $(12)$ and $(12\ldots n)$, since those two permutations generate $S_n$. In this case, we already have $(12345)\in\Autclif(C)$, so we try $(12)$. This permutation sends the generators to
\begin{align*}
&ZXZXI\\
&XIZZX\\
&IXXZZ\\
&XZIXZ,
\end{align*}
and twisting these by the permutations $(YZ)$, $(YZ)$, $(XZ)$, $(XY)$, and $(XZ)$ in the respective tensor factors gives
\begin{align*}
&YXXYI\\
&XIXZZ\\
&IXZZX\\
&XYIYX,
\end{align*}
all of which are in $S$. Hence $(12)\in\Autclif(C)\Ra\Autclif(C)=S_5$.

\subsubsection{A $[[6,0,4]]$ code}\label{sec2.2.2}

Recall the $\mathbf{0}_L$ and $\mathbf{1}_L$ states from Example \ref{513Code1}, which form a basis for the $[[5,1,3]]$ code mentioned above. We may consider the state $\mathbf{0}\otimes\mathbf{0}_L+\mathbf{1}\otimes\mathbf{1}_L\in\C^{2^6}$, which defines a (degenerate) $[[6,0,4]]$ code $C$. This ``code" is important as an example of a \emph{maximally entangled state}, following the terminology of \cite{Harvey2020}.

It is not hard to check that the stabilizing subgroup $S\subset G_6$ corresponding to $C$ is generated by
\begin{align*}
&IXZZXI\\
&IIXZZX\\
&IXIXZZ\\
&IZXIXZ\\
&XXXXXX\\
&ZZZZZZ.
\end{align*}
It turns out that while $\Autstr(C)$ is again $\gen{(23456),(36)(45)}\cong D_{10}$, $\Autweak(C)$ is in fact $\gen{(23456),(135)(264)}\cong\PSL(2,5)\cong A_5$. For instance, $(135)(264)$ sends $IXZZXI$ to $XZIIZX$, which is not in $S$, but rather in $-S$.

Note that this is another example where $\Autstr(C)$ is not normal in $\Autweak(C)$.

To calculate these groups, we adopt the procedure from Section \ref{sec2.2.1} and write out all the $2^{6-0}=64$ elements of $S$. Since any permutation fixes $XXXXXX$ and $ZZZZZZ$, we can focus on the first four generators of $S$. It is a good strategy to do casework on the image of the first tensor factor, since we would then need to find four elements in $S$ with $I$'s in that slot (perhaps with the proper sign, depending on whether we are calculating the strong or weak automorphism group). We also use the constraint that the first four generators must be sent to elements of $S$ (or $-S$) that have 2 $I$'s, 2 $X$'s, and 2 $Z$'s as tensor factors. All of these types of criteria---location of the $I$'s in the tensor product, weights of elements, and the types of Pauli matrices used in each tensor product---are highly useful when performing hand computations.

Finally, we claim that $\Autclif(C)=S_6$. First, we check that $(12)\in\Autclif(C)$: it sends the generators to
\begin{align*}
&XIZZXI\\
&IIXZZX\\
&XIIXZZ\\
&ZIXIXZ\\
&XXXXXX\\
&ZZZZZZ,
\end{align*}
and twisting these by the permutations $(XY)$, $(XY)$, $(XZ)$, $(YZ)$, $(YZ)$ and $(XZ)$ in the respective tensor factors gives
\begin{align*}
&YIXYXI\\
&IIZYYZ\\
&YIIXYX\\
&ZIZIXX\\
&YYZXXZ\\
&ZZXYYX,
\end{align*}
all of which are in $S$. It remains to check that $(123456)\in\Autclif(C)$. This sends the generators to
\begin{align*}
&IIXZZX\\
&XIIXZZ\\
&ZIXIXZ\\
&ZIZXIX\\
&XXXXXX\\
&ZZZZZZ,
\end{align*}
and twisting these by the permutations $(XY)$, $(XY)$, $(XZ)$, $(YZ)$, $(YZ)$ and $(XZ)$ in the respective tensor factors gives
\begin{align*}
&IIZYYZ\\
&YIIXYX\\
&ZIZIXX\\
&ZIXXIZ\\
&YYZXXZ\\
&ZZXYYX,
\end{align*}
all of which are in $S$. Again, for this computation, it is important to exploit the position of $I$'s as tensor factors, since conjugation by elements of $L_6$ cannot change $I$.

\subsubsection{The $[[7,1,3]]$ Steane code}\label{sec2.2.3}

The $[[7,1,3]]$ Steane code $C$ corresponds to a subgroup $S\subset G_7$ that has generators given by
\begin{align*}
&IIIXXXX\\
&IXXIIXX\\
&XIXIXIX\\
&IIIZZZZ\\
&IZZIIZZ\\
&ZIZIZIZ.
\end{align*}
$C$ is derived from the classical $[7,4]$ Hamming code---this is apparent if one writes out the check matrix representations of these generators, and compares that matrix to the parity-check matrix of the Hamming code. It is of primary interest as an example of a \emph{CSS code} (see Section 10.4.2 of \cite{Nielsen2011} for the full details). Since it is derived from the Hamming code, which has an automorphism group of $\PGL(3,2)\cong\PSL(2,7)$, it should be of no surprise that $\Autstr(C)=\Autweak(C)=\gen{(46)(57),(124)(365)}\cong\PGL(3,2)$. In general, it is at least true that the strong automorphism group of a CSS code is isomorphic to the automorphism group of the classical code it is derived from.

By writing out all of the $2^{7-1}=64$ elements of $S$, one can check that $\Autclif(C)$ is also $\PGL(3,2)$. This provides a nontrivial example where all three types of automorphism group are equal.

\subsubsection{An $[[8,3,3]]$ code}\label{sec2.2.4}

The tables at \cite{codetables} give us the generators for the stabilizing subgroup $S\subset G_8$ of an $[[8,3,3]]$ code $C$:
\begin{align*}
&XIZIYZXY\\
&IXZZYXYI\\
&IZXIYYZX\\
&IZIYZXXY\\
&ZZZZZZZZ.
\end{align*}
This code is significant because it is the smallest-sized code encoding 3 logical qubits that can correct any single-qubit error. We can calculate that
\begin{align*}
\Autstr(C)=\{&(1),(12)(35)(68)(47),(13)(25)(48)(67),(14)(27)(38)(56),(15)(23)(46)(78),\\&(16)(28)(37)(45),(17)(24)(36)(58),(18)(26)(34)(57)\}\cong\Z_2\times\Z_2\times\Z_2,
\end{align*}
\begin{equation*}
\Autweak(C)=\Autstr(C)\rtimes\gen{(2453876)}\cong(\Z_2\times\Z_2\times\Z_2)\rtimes\Z_7.
\end{equation*}
We may alternatively identify $\Autweak(C)$ as $\AGL(1,8)$, the 1-dimensional affine general linear group over $\F_8$, which acts sharply 2-transitively on 8 points. This is a group of order 56 with a unique Sylow 2-subgroup isomorphic to $\Z_2\times\Z_2\times\Z_2$.

To calculate these groups, one can use the following remarkable description of the elements of $S$: each of the $2^{8-3}-4=28$ elements in $S$ that are not $IIIIIIII$, $XXXXXXXX$, $ZZZZZZZZ$, or $YYYYYYYY$, contain exactly 2 $I$'s, 2 $X$'s, 2 $Z$'s, and 2 $Y$'s as tensor factors. Moreover, for any pair of integers $1\leq i<j\leq 8$, exactly one of those 28 elements has $I$'s as tensor factors in slots $i$ and $j$.

To calculate $\Autclif(C)$, we use the fact that $\Autclif(C)\supseteq\Autweak(C)$ is 2-transitive, so it suffices to find the stabilizer of slots 1 and 2. A tedious check by casework (using the positions of $I$ tensor factors, as well as the above observations about $S$, as our guide) shows that this stabilizer is of order 3 and generated by $(367)(458)$. Hence $\Autclif(C)$ is a group of order $56\cdot3=168$. We may identify $\Autclif(C)$ as the 1-dimensional affine \textit{semilinear} group $\AGamL(1,8)$, which is a semidirect product $\AGL(1,8)\rtimes\Aut(\F_8)$.

\subsubsection{An $[[8,2,3]]$ code}\label{sec2.2.5}

It turns out that the $[[8,3,3]]$ code mentioned in \ref{sec2.2.4} has an $[[8,2,3]]$ subcode $C$, which is the smallest-sized code encoding 2 logical qubits that can correct any single-qubit error. Its stabilizing subgroup $S$ is generated by
\begin{align*}
&XIZIYZXY\\
&IXZZYXYI\\
&IZXIYYZX\\
&IZZYZXZZ\\
&IIZIIIYX\\
&ZZZZZZZZ.
\end{align*}
Note that the product of the fourth and fifth generators of $S$ give the fourth generator of the aforementioned $[[8,3,3]]$ code.

Now, any element of $\Autstr(C)$ or $\Autstr(C)$ sends the generator $IIZIIIYX$ to another weight-3 element in $S$ (up to sign). By analyzing the weight-3 elements in $S$ (there are only 8 of them), we can conclude that the only nontrivial element in both $\Autstr(C)$ and $\Autweak(C)$ is $(13)(25)(48)(67)$. So although the strong automorphism group of the above $[[8,3,3]]$ code was sharply 1-transitive, and the weak automorphism group of the same code was sharply 2-transitive, neither $\Autstr(C)\cong\Z_2$ nor $\Autweak(C)=\Autstr(C)$ are transitive.

\subsubsection{A $[[10,0,4]]$ code}\label{sec2.2.6}

Consider the following generators of a stabilizing subgroup $S\subset G_{10}$:
\begin{align*}
&XIIZXZXZII\\
&IXIIZXZXZI\\
&IIXIIZXZXZ\\
&ZIIXIIZXZX\\
&XZIIXIIZXZ\\
&ZXZIIXIIZX\\
&XZXZIIXIIZ\\
&ZXZXZIIXII\\
&IZXZXZIIXI\\
&IIZXZXZIIX
\end{align*}

These ten elements are cyclic permutations of each other. It can be checked that $S$ determines a degenerate $[[10,0,4]]$ code $C$, and from \cite{codetables} we see that 4 is the best possible distance for a $[[10,0]]$ code. Note that this is not the same $[[10,0,4]]$ code given at \cite{codetables}.

Now, we will later see (Theorems \ref{M12thm} and \ref{M11thm}) that there are no nondegenerate 12-qubit (resp. 11-qubit) stabilizer codes with strong or weak automorphism group $M_{12}$ (resp. $M_{11}$). This suggests that the Clifford-twisted automorphism group might be the ``correct" automorphism group to investigate, if we are trying to find some connections between stabilizer codes and the small Mathieu groups. This code provides some evidence of such a connection: we have
\[
\Autclif(C)=\gen{(1~2~3~4~5~6~7~8~9~10),(8~9)(4~10)(5~6)}\cong M_{10}.2,
\]
a 3-transitive group of order 1440. This group is also isomorphic to the 2-dimensional projective semilinear group $\PGamL(2,9)$, as well as to the automorphism group of the unique (up to isomorphism) $(3,4,10)$ Steiner system.

It is evident that $(1~2~3~4~5~6~7~8~9~10)\in\Autclif(C)$, since it is even a strong automorphism. To see that $(8~9)(4~10)(5~6)\in\Autclif(C)$, one applies this permutation to the above ten generators and then twists by $(XZ)$ in the 5th, 6th, 8th, and 9th tensor factors (in the other tensor factors, we do nothing). Recall from the discussion surrounding Definition \ref{cliffordGroup} that this $(XZ)$ twist is, up to sign, the Hadamard gate. The resulting elements of $G_{10}$ will all be in $S$. Since $M_{10}.2$ is a maximal subgroup of $S_{10}$, it remains to show that some transposition, say $(13)$, is not in $\Autclif(C)$. This is not too hard to do using the program in Appendix \ref{SAGEcode}: one way is to apply $(13)$ to the first and third generators listed above (so they become $IIXZXZXZII$ and $XIIIIZXZXZ$, respectively), and show that the resulting two elements of $G_{10}$ cannot be simultaneously twisted into elements of $S$.

The general heuristic for the construction of this code is as follows: somehow the code should be related to the unique $(3,4,10)$ Steiner system, since as mentioned above $M_{10}.2$ is the automorphism group of that Steiner system. From \cite{SteinerWolfram} we obtain a description of the blocks of such a system, and we place $I$'s in tensor factors corresponding to ten of those blocks, which will be cyclic permutations of each other. In some way, this makes sense to try, since the identity matrices act as ``distinguished elements": $X$'s, $Y$'s, and $Z$'s can all be twisted into each other, but that is not true with $I$'s. The rest of the construction consists of playing around with the remaining six tensor factors and eventually coming upon the ``right ones".

As a final remark, the Clifford-twisted automorphism group of this code is indeed the most interesting: both the strong and weak automorphism groups of $C$ are \\$\gen{(1~2~3~4~5~6~7~8~9~10),(2~10)(3~9)(4~8)(5~7)}\cong D_{20}$.

\subsubsection{An $[[11,1,5]]$ code}\label{sec2.2.7}

Table 8.5 of \cite{Gottesman1997} gives the generators for the stabilizing subgroup $S$ of an $[[11,1,5]]$ code $C$ as follows:
\begin{align*}
&ZZZZZZIIIII\\
&XXXXXXIIIII\\
&IIIZXYYYYXZ\\
&IIIXYZZZZYX\\
&ZYXIIIZYXII\\
&XZYIIIXZYII\\
&IIIZYXXYZII\\
&IIIXZYZXYII\\
&ZXYIIIZZZXY\\
&YZXIIIYYYZX
\end{align*}
This code is significant because it is the smallest-sized code that can correct arbitrary errors on any two qubits (see \cite{codetables}). $S$ is slightly too large to analyze using the aforementioned techniques, or via brute-force using the program in Appendix \ref{SAGEcode}; however, we strongly suspect that both $\Autstr(C)$ and $\Autweak(C)$ are trivial.
\begin{question}\label{q2}
Are the strong and weak automorphism groups of $C$ trivial?
\end{question}

\section{Stabilizer Codes with Highly Transitive Automorphism Groups}\label{sec3}

In this section, we try to investigate the general question of ``\emph{how symmetric can a good stabilizer code be?}" We will focus on stabilizer codes that have distance at least 3 (so they can at least correct one error on any qubit, hence ``good"), and we will try to find those codes with multiply transitive strong and weak automorphism groups (hence ``symmetric").

The examples in Section \ref{sec2.2} suggest that in general, the strong and weak automorphism groups of a good code cannot be ``too symmetric". For instance, the highest degree of transitivity we observed was 2-transitive, such as in the $[[7,1,3]]$ Steane code with strong and weak automorphism group $\PGL(3,2)$ (Section \ref{sec2.2.3}), and in a $[[8,3,3]]$ code with weak automorphism group $\AGL(1,8)$ (Section \ref{sec2.2.4}). We believe that this is a good heuristic, and we now provide some results supporting this point of view.

\begin{theorem}\label{SnAutGroup}
An $[[n,k,d]]$ stabilizer code $C$ with $S_n$ or $A_n$ as its strong or weak automorphism group has $d\leq2$.
\end{theorem}

We make a useful definition:
\begin{definition}\label{complexityDef}
Let $T$ be a subset of $E_n$. Then the \textbf{complexity} of $T$ is the number of different Pauli matrices ($I$, $X$, $Y$, or $Z$) that appear as tensor factors in elements of $T$ (so the complexity is an integer in $\{1,2,3,4\}$). We define the complexity of an element similarly. For example, the subset $T=\{IXX, YYZ\}$ has complexity 4, even though both elements $IXX$ and $YYZ$ have complexity 2.
\end{definition}

\begin{proof}[Proof of Theorem \ref{SnAutGroup}]
Let $C$ be the stabilizer code in question with parameters $[[n,k,d]]$, and let $S\subset G_n$ be its stabilizing subgroup. Note that any stabilizer code with distance at least 3 must be able to correct an arbitrary error on a single qubit, which means it must be at least 5 physical qubits large (see, for example, the quantum Hamming bound discussed in Section 10.3.4 of \cite{Nielsen2011}). So we may assume that $n\geq5$.

We prove the theorem in the case $\Autweak(C)\cong S_n$ or $A_n$, and the proof for $\Autstr(C)$ is essentially the same.

Now, $S$ has complexity 1, 2, 3, or 4. If $S$ has complexity 1, then it is the trivial subgroup, so we can disregard this case. If $S$ has complexity 2, then without loss of generality every element $s\in S$ is a tensor product of $I$'s and $X$'s. Then $XI\ldots I$ is in the normalizer $N(S)$, so either it is in $S$ or $d=1$. The same applies for $IXI\ldots I$, $IIXI\ldots I$, etc., so either $S$ contains them all or $d=1$. In the former case, $S$ will have size $2^n$, so that $k=0$. In this case, $d$ still equals 1, applying the distance convention for zero-dimensional stabilizer codes as discussed in Remark \ref{distk=0}.

If $S$ has complexity 3, then without loss of generality, every element $s\in S$ is a tensor product of $I$'s, $X$'s, and $Z$'s. Suppose $s_1\in S$ has an $X$ in slot $i$ (the $i$th tensor factor) and $s_2\in S$ has a $Z$ in slot $j$. Then we may choose some permutation $\sigma\in\Autweak(C)$, which is either $S_n$ or $A_n$, such that $\sigma(i)=j$. So plus or minus $\sigma(s_1)$ is in $S$, and taking the positive sign (without loss of generality), we see that $\sigma(s_1)s_2\in S$. But this element has a $Y$ in slot $j$, contradicting the complexity of $S$.

If $S$ has complexity 4, we first claim that no element in $S$ can have complexity 3 or 4. If $s\in S$ had complexity 3 or 4, then it has $I$, $X$, and $Z$ as tensor factors somewhere, say \\$s=\ldots I\ldots X\ldots Z\ldots$ (without loss of generality, since it will be apparent that it does not matter where these tensor factors are located). Then applying an appropriate 3-cycle $\sigma\in A_n\subseteq\Autweak(C)$, we get $\sigma(s)=\ldots Z\ldots I\ldots X\ldots$, so that plus or minus $\sigma(s)$ is in $S$. In either case, the point is that $\sigma(s)$ and $s$ do not commute, a contradiction.

So elements in $S$ have complexity at most 2. Suppose we had some complexity 2 element $s$, which is, without loss of generality, a tensor product of $I$'s and $X$'s only, with at least 1 copy of each. Let $s$ have an $I$ in slot $i$ and an $X$ in slot $j$. Now, because we assume $n\geq5$, there are at least 3 more slots, so either $I$ or $X$ appears in the remaining slots at least twice. Without loss of generality, suppose $I$ appears in slots $k$ and $l$ of $s$, where $i, j, k, l$ are pairwise distinct. Consider the permutation $\sigma=(ij)(kl)\in A_n\subseteq\Autweak(C)$, so that $\sigma(s)$ has the effect of only switching the $I$ in slot $i$ and the $X$ in the slot $j$ (the $I$'s in slots $k$ and $l$ are invisibly switched). Then plus or minus $\sigma(s)$ is in $S$, so taking the positive sign, we see that $t\coloneqq\sigma(s)s\in S$ has an $I$ in every slot of the tensor product, except in slots $i$ and $j$, where it has $X$'s.

Then because $n$ is at least 5, by applying elements in $\Autweak(C)$, it follows that any tensor product of $I$'s and exactly two $X$'s is in $S$, up to sign. Then the elements \\$\{\pm XXI\ldots I, \pm IXXI\ldots I,\ldots, \pm I\ldots IXX\}$ are independent in $S$, so $S$ has size at least $2^{n-1}$. But the complexity 4 hypothesis means that $S$ contains some term with a $Y$ or $Z$, and that element is certainly independent from the above set, so $S$ must have size exactly $2^n$. So $k=0$, and by the distance convention for zero-dimensional stabilizer codes, $d\leq2$.

The only case we have not considered is where all elements in $S$ have complexity 1. Then $S$ is a subgroup contained in $\{\pm I\ldots I, \pm X\ldots X, \pm Z\ldots Z, \pm Y\ldots Y\}$, and $XXI\ldots I$ is a weight 2 element in $N(S)-S$, so $d=2$.
\end{proof}

The classification of finite simple groups tells us that the only $k$-transitive groups for $k\geq6$ are $S_n$ and $A_n$ for large enough $n$. We also know that the only 5-transitive groups are the Mathieu groups $M_{24}$ and $M_{12}$, and the only 4-transitive groups are $M_{23}$ and $M_{11}$. Therefore it makes sense to discuss these groups next. We will now deal with the small Mathieu groups, starting with $M_{12}$.

\begin{theorem}\label{M12thm}
There is no $k\geq1$ stabilizer code with ambient space $\C^{2^{12}}$ that has strong or weak automorphism group $M_{12}$.
\end{theorem}

\begin{proof}
We prove the statement in the $\Autstr$ case, since the $\Autweak$ case only gives us increased freedom to make sign changes in elements of the stabilizer group, which doesn't make a difference. In particular, in all that follows, we will mean ``equal up to a sign $\pm$" when we say ``equal", unless stated (we will only need to deal with signs in a few places). The only fact about $M_{12}$ we will need is that it is 5-transitive as a permutation group on 12 points.

Let $C$ be a stabilizer code (encoding at least 1 logical qubit) corresponding to the stabilizer subgroup $S\subset E_{12}$ with $\Autstr(C)\supseteq M_{12}$. We prove that this inclusion must be strict.

Disregarding the trivial case, the complexity of $S$ must be 2, 3, or 4. First, if the complexity of $S$ is 2, then without loss of generality, all elements are tensor products of $X$'s and $I$'s. Pick some $s\in S$ not equal to $I$ or $X\ldots X$ (if there is no such element, then $\Autstr(C)\cong S_{12}$), so that $s$ has at most 6 $X$'s as tensor factors. Note that we could change the $X$'s and $I$'s and the argument still works. We can assume without harm that the $X$'s are located consecutively in the first slots, so $s$ is one of the following, up to sign:
$$\{XI\ldots I, XXI\ldots I, XXXI\ldots I, XXXXI\ldots I, XXXXXI\ldots I, XXXXXXIIIIII\}.$$
It does not actually matter where the $X$'s are located, since the argument can be easily modified for any position of the $X$'s.

Because $M_{12}$ is 5-transitive, we may produce one of the elements in the below list by applying a permutation from $\Autstr(C)$ to $s$. The dots mean that we do not exactly know the order of the remaining tensor factors, since we may only specify the images of 5 points (qubits).
$$\{IXI\ldots, XIXI\ldots, XXIXI\ldots, XXXIX\ldots, XXXXI\ldots, XXXXI\ldots\}.$$
Let $s'\in S$ be the corresponding element to $s$. Notice that in each case, $ss'$ is an element with weight 2, except in the last case, where $ss'$ could have weight 4. But in the last case we can just repeat this procedure to produce some element of weight 2. Let us call this process \textbf{reduction to a weight 2 element}, since we will want to refer to it later. We have essentially shown that if $S$ contains an element of weight at most 6, then it contains an element of weight 2. %This weight 2 element must also have complexity 2: an element like $XZI\ldots$ cannot be in $S$, as otherwise applying an appropriate automorphism would give $IXZI\ldots\in S$, which does not commute with $XZI\ldots$.

So $S$ contains some element of weight 2. By 5-transitivity of $\Autstr(C)\supseteq M_{12}$, $S$ contains all tensor products of $I$'s and $X$'s of weight 2, of which the 11 elements \\$\{XXI\ldots I, IXXI\ldots I,\ldots, I\ldots IXX\}$ are independent. Since $k\geq1$, $S$ is generated by exactly these elements. Now, notice that if we apply the appropriate element in the 5-transitive group $\Autstr$ to $IXXI\ldots I$, we get $XIXI\ldots I\in S$ \emph{with the same sign as $IXXI\ldots I$}. Multiplying these, we see that $+XXI\ldots I\in S$ (with the positive sign added for emphasis). Then $XIXI\ldots I\in S$, the product of the first two generators, has the same sign as $IXXI\ldots I\in S$, so the transposition $(12)$ is in $\Autstr$ (it fixes the rest of the generators). Therefore $\Autstr(C)\not=M_{12}$, since $M_{12}$ does not contain any transposition. This finishes the discussion of the complexity 2 case.

\iffalse
$k=0$ preliminary work:

Now, $S$ has at most 1 more independent generator, since $S$ has size at most $2^{12}$. There are 2 cases:
\begin{enumerate}
\item $S$ is generated by the elements in the above list. We know that $XXI\ldots I$ has positive sign (that is, $XXI\ldots I\in S$). Then $XIXI\ldots I\in S$, the product of the first two generators, has the same sign as $IXXI\ldots I\in S$, so the transposition $(12)$ is in $\Autstr$ (it fixes the rest of the generators). Therefore $\Autstr\not=M_{12}$, since $M_{12}$ does not contain any transposition.
\item $S$ has 1 more independent generator $t$. Because the elements in \\$\{XXI\ldots I, IXXI\ldots I,\ldots, I\ldots IXX\}$ can generate any tensor product of $I$'s and $X$'s of even weight, $t$ must have odd weight (otherwise $-I\in S$ or $t$ is not independent from the above generators). But $S$ contains all tensor products of $I$'s and $X$'s of weight 2, which along with $t\in S$, implies that there is some element of weight 1 in $S$. The large strong automorphism group of $S$ implies that $S$ contains the independent elements $\{XI\ldots I, IXI\ldots I,\ldots, I\ldots IX\}$, which generate $S$. A similar argument as the one in the previous paragraph shows that all of these elements must have the same sign, in which case $\Autstr(C)\cong S_{12}\not=M_{12}$.
\end{enumerate}
\fi

As in the proof regarding $S_n$ or $A_n$ automorphism group, transitivity of $\Autstr(C)$ shows that $S$ cannot have complexity 3.

It remains to discuss the complexity 4 case. First, we show that there cannot be an element in $S$ with complexity 2. This argument essentially follows the above work where $S$ has complexity 2: given an element $s\in S$ with complexity 2, we may reduce to a weight 2 element $t\in S$ that only has $I$'s and one other Pauli matrix as tensor factors (say, $X$). Then the same argument shows that $S$ would have to be generated by $\{XXI\ldots I, IXXI\ldots I,\ldots, I\ldots IXX\}$ (by the $k\geq1$ hypothesis). This is a contradiction to $S$ being complexity 4.

So there is an element in $S$ with complexity at least 3 (if all elements in $S$ had complexity 1, then $\Autstr(C)\cong S_{12}\not=M_{12}$). We claim that there is in fact an element in $S$ with complexity 4. Suppose the first three slots of $s\in S$ are $XIZ$ (as usual, it doesn't matter what order they are in, what slots they are in, or what Pauli matrices they are, as long as they are pairwise distinct). There are four possibilities:
\begin{enumerate}
\item $s$, interpreted as a string of Pauli matrices, begins $XIZI\ldots$. By 5-transitivity of $\Autstr(C)$ there is some $s'\in S$ that begins $ZXII\ldots$. Then their product $ss'$ begins $YXZI\ldots$, the desired element of complexity 4.
\item $s$ begins $XIZX\ldots$. By 5-transitivity of $\Autstr(C)$ there is some $s'\in S$ that begins $IZXX\ldots$. Then their product $ss'$ begins $XZYI\ldots$, the desired element of complexity 4.
\item $s$ begins $XIZZ\ldots$. By 5-transitivity of $\Autstr(C)$ there is some $s'\in S$ that begins $IZXZ\ldots$. Then their product $ss'$ begins $XZYI\ldots$, the desired element of complexity 4.
\item $s$ begins $XIZY\ldots$. This already has complexity 4.
\end{enumerate}
Therefore we may assume without loss of generality that $s$ has complexity 4, and that the first four tensor factors of $s$ are $XZYI$.

Now, we look at elements of $N(S)-S$. By the bounds found at \cite{codetables}, $N(S)-S$ contains an element of weight at most 6. We turn to looking at the possibilities case by case:

\begin{enumerate}
\item Weight 1:
\begin{enumerate}
\item $XI\ldots I\not\in N(S)$ because it does not commute with $ZIXY\ldots$. A similar argument shows that $YI\ldots I$, $ZI\ldots I\not\in N(S)$.
\end{enumerate}
\item Weight 2:
\begin{enumerate}
\item $XXI\ldots I\not\in N(S)$ because it does not commute with $ZIXY\ldots$. A similar argument shows that $YYI\ldots I$, $ZZI\ldots I\not\in N(S)$.
\item $XZI\ldots I\not\in N(S)$ because it does not commute with $ZIXY\ldots$. A similar argument shows that $ZXI\ldots I, XYI\ldots I, YXI\ldots I, YZI\ldots I, ZYI\ldots I\not\in N(S)$.
\end{enumerate}
\item Weights 3 through 6: it is easily seen that if $g\in N(S)$, then $\sigma(g)\in N(\sigma(S))=N(S)$ for $\sigma\in\Autstr(C)$ (the same argument works for $\Autweak(C)$, since in that case $\sigma(S)$ only differs from $S$ by signs). Then if $g\in N(S)$ has weight 3, 4, 5, or 6, we may reduce to an element $g'\in N(S)$ of weight 2. This falls into the above case.
\end{enumerate}

In all cases, we get a contradiction, so we are done.

\end{proof}

%\begin{remark}\label{rmkA}
%Still working on the $k=0$ case. This is more difficult because there is additional casework, but more importantly, it is possible for a $[[12,0]]$ code to have distance $6$. So the last argument would have to be modified, which poses problems since we can only control 5 tensor factors when applying an automorphism.
%\end{remark}

\begin{remark}\label{rmkB}
This method fails with $M_{23}$ and $M_{24}$. We know that there is a $[[23,1,7]]$ stabilizer code that at least has strong automorphism group $M_{23}$; this is the CSS code generated from the classical $[23,12,7]$ Golay code (see section 7.15.4 of \cite{Preskillnotes}. There should also be a $[[24,0,8]]$ ``code" that is the CSS code generated from the classical $[24, 12, 8]$ extended Golay code.%There are many reasons for this: 5-transitivity is ``too weak" to control enough qubits in a 24-element tensor product, a 24-qubit stabilizer code encoding 1 logical qubit could have distance as large as 9, etc.
\end{remark}

%\begin{remark}\label{rmkC}
%Unfortunately, a similar argument does not work with $\C^{2^{11}}$ and $M_{11}$, the problem being that $M_{11}$ is ``only" 4-transitive, so we cannot control enough elements in the final casework (looking at elements of weight 1,2,3,4,5, and we will fail at weight 5 elements).
%\end{remark}

%Despite the above remark, we can still show the following:

One can use the same method to show the following:

\begin{theorem}\label{M11thm}
There is no $k\geq1$ stabilizer code with ambient space $\C^{2^{11}}$ that has strong or weak automorphism group $M_{11}$.
\end{theorem}

The main difference is that reduction to a weight 2 element is now only valid for elements of weight at most 5. But luckily this does not pose a problem: $\floor{\frac{11}{2}}=5$, so the arguments for complexity 2 elements still hold, and the bounds at \cite{codetables} show that $N(S)-S$ for any stabilizing subgroup $S\subset G_{11}$ contains an element of weight at most 5, so the last casework step can be repeated. Every other step in the proof of Theorem \ref{M12thm} can be done using only 4-transitivity.\\

We end this paper by presenting some questions that naturally follow from topics discussed above.

\begin{question}\label{q3}
Are Theorems \ref{M12thm} and \ref{M11thm} true if the $k\geq1$ hypothesis is dropped?
\end{question}

\begin{question}\label{q4}
Describe stabilizer codes $C$ with $\Autclif(C)\cong S_n$. Replace $S_n$ by $A_n$ and the Mathieu groups, if possible.
\end{question}

\begin{question}\label{q5}
Describe stabilizer codes $C$ with 3-transitive $\Autstr(C)$ or $\Autweak(C)$. Replace ``3-transitive" by ``2-transitive", ``transitive", and ``cyclic" if possible.
\end{question}

\begin{question}\label{q6}
Can a stabilizer code with trivial strong (resp. weak) automorphism group be arbitrarily good? That is, are there stabilizer codes with trivial strong (resp. weak) automorphism groups having arbitrarily large distance? Replace ``trivial" with ``cyclic" if possible.
\end{question}

\bibliography{Investigations_on_Automorphism_Groups_of_Quantum_Stabilizer_Codes}

\begin{thebibliography}{10}

\bibitem{Bolt1961}
Beverley Bolt, T.~G. Room, and G.~E. Wall.
\newblock On the clifford collineation, transform and similarity groups. i.
\newblock {\em Journal of the Australian Mathematical Society}, 2(1):60–79,
  1961.

\bibitem{Bolt1961-2}
Beverley Bolt, T.~G. Room, and G.~E. Wall.
\newblock On the clifford collineation, transform and similarity groups. ii.
\newblock {\em Journal of The Australian Mathematical Society}, 2:80--96, 1961.

\bibitem{Calderbank1997}
A.R. Calderbank, E.M. Rains, P.W. Shor, and N.J.A. Sloane.
\newblock Quantum error correction via codes over {G}{F}(4).
\newblock pages 292--, 1997.

\bibitem{Gottesman1997}
Daniel Gottesman.
\newblock {\em Stabilizer codes and quantum error correction}.
\newblock PhD thesis, California Institute of Technology, May 1997.

\bibitem{Gottesman2009}
Daniel Gottesman.
\newblock An introduction to quantum error correction and fault-tolerant
  quantum computation.
\newblock 2009.

\bibitem{codetables}
Markus Grassl.
\newblock Quantum error-correcting code tables.
\newblock \url{http://codetables.de}, 2019.
\newblock Accessed: 2021-08-08.

\bibitem{Harvey2020}
Jeffrey~A. Harvey and Gregory~W. Moore.
\newblock Moonshine, superconformal symmetry, and quantum error correction.
\newblock {\em Journal of High Energy Physics}, 2020(5), May 2020.

\bibitem{Nielsen2011}
Michael~A. Nielsen and Isaac~L. Chuang.
\newblock {\em Quantum Computation and Quantum Information: 10th Anniversary
  Edition}.
\newblock Cambridge University Press, USA, 10th edition, 2011.

\bibitem{Preskillnotes}
John Preskill.
\newblock Chapter 7 of notes on quantum computation.
\newblock \url{http://theory.caltech.edu/~preskill/ph229/notes/chap7.pdf}.

\bibitem{Shor1995}
Peter~W. Shor.
\newblock Scheme for reducing decoherence in quantum computer memory.
\newblock {\em Phys. Rev. A}, 52:R2493--R2496, Oct 1995.

\bibitem{SteinerWolfram}
Eric~W. Weisstein.
\newblock Steiner quadruple system.
\newblock \url{https://mathworld.wolfram.com/SteinerQuadrupleSystem.html}.

\end{thebibliography}
\bibliographystyle{plain}

\newpage

\begin{appendices}
\section{SAGE Code for Computing Automorphism Groups}\label{SAGEcode}

Here is some SAGE code using a brute-force method to calculate the strong and weak automorphism groups of a stabilizer code. Thanks to Dr. Daniel Bump for his major assistance in writing this program.

\begin{verbatim}
"""
SAGE Code for looking for automorphisms of stabilizer quantum error
correcting codes.
"""

# For reference, x,y,z = Pauli matrices
# h,s = Hadamard and Phase gates.

X = Matrix([[0,1],[1,0]])
Y = Matrix([[0,-i],[i,0]])
Z = Matrix([[1,0],[0,-1]])
H = Matrix([[1,1],[1,-1]])
S = Matrix([[1,0],[0,i]])

# The Pauli error group is defined by generators and relations:

FG.<x,y,z,j> = FreeGroup()
PE = FG/{j^4,x^2,y^2,z^2,x*y/x/y*j^2,y*z/y/z*j^2,z*x/z/x*j^2,x*j/x/j,y*j/y/j, 
z*j/z/j,x*y/j/z,y*x*j/z,y*z/j/x,z*y*j/x,z*x/j/y,x*z*j/y}
PE.inject_variables()

def ca(a, split=False):
    """
    Canonical form for the Pauli Error group elements
    """
    for b in [PE.one(),x,y,z]:
        for k in [PE.one(),j,j^2,j^3]:
            if a == k*b:
                if split:
                    return [k,b]
                else:
                    return k*b
    return a                

def emul(A,B,debug=False):
    """
    Multiplication for Pauli error group elements in canonical form
    """
    a0 = A[0]
    b0 = B[0]
    r0 = a0*b0
    rt = []
    for [t,u] in zip(A[1:],B[1:]):
        [p,q]=ca(t*u,split=True)
        r0 *= p
        if debug:
            print (t,u,p)
        rt.append(q)
    ret = [ca(r0)]
    for t in rt:
        ret.append(t)
    return ret

def eprod(M,n):
    """
    M is a subset of the stabilizer group S
    Returns the product of the elements of M.
    """
    ret = tuple((n+1)*[PE.one()])
    for m in M:
        ret = emul(ret,m)
    return list(ret)

def unpack(st):
    """
    Pauli error group elements may be represented by strings
    and unpacked by this function.

    sage: unpack("XYZZY")
    (1, x, y, z, z, y)
    """
    ret = [PE.one()]
    for c in st:
        if c == "X":
            ret.append(x)
        elif c == "Y":
            ret.append(y)
        elif c == "Z":
            ret.append(z)
        elif c == "I":               
            ret.append(PE.one())
    return tuple(ret)

def pack(w):
    ret = ""
    for t in w[1:]:
        if t==PE.one():
            ret += "I"
        elif t==x:
            ret += "X"
        elif t==y:
            ret += "Y"
        elif t==z:
            ret += "Z"
    return ret

def Stabilizer(S):
    """
    Return the stabilizer group generated by a set of generators
    of the Pauli error group in string notation. The generators
    must commute and the generated group may not contain -I.
    """
    n = len(S[0])
    return [eprod(M,n) for M in Set(unpack(s) for s in S).subsets()]

def GenToArray(S):
    """
    Takes an array of generator strings and returns the same Pauli elements
    in standard form.
    """
    return [list(unpack(s)) for s in S]

StabTest = ["XXX", "YYI", "ZXZ"]
Stab513 = ["XZZXI","IXZZX","XIXZZ","ZXIXZ"] 
Stab604 = ["IXZZXI","IIXZZX","IXIXZZ","IZXIXZ","XXXXXX","ZZZZZZ"]
Stab713 = ["IIIXXXX", "IXXIIXX", "XIXIXIX", "IIIZZZZ", "IZZIIZZ", "ZIZIZIZ"]
Stab833 = ["XIZIYZXY", "IXZZYXYI", "IZXIYYZX", "IZIYZXXY", "ZZZZZZZZ"]
Stab823 = ["XIZIYZXY", "IXZZYXYI", "IZXIYYZX", "IZZYZXZZ", "IIZIIIYX", "ZZZZZZZZ"]
Stab933 = ["XIZIYZXYI", "IXZZYXYII", "IZXIYYZXI", "IZIYZXXYI", "ZZZZZZZZI",
 "IIIIIIIIX"]
Stab1004 = ["XIIZXZXZII", "IXIIZXZXZI", "IIXIIZXZXZ", "ZIIXIIZXZX", "XZIIXIIZXZ",
 "ZXZIIXIIZX", "XZXZIIXIIZ", "ZXZXZIIXII", "IZXZXZIIXI", "IIZXZXZIIX"]
Stab1115 = ["ZZZZZZIIIII", "XXXXXXIIIII", "IIIZXYYYYXZ", "IIIXYZZZZYX",
 "ZYXIIIZYXII", "XZYIIIXZYII", "IIIZYXXYZII", "IIIXZYZXYII", "ZXYIIIZZZXY",
 "YZXIIIYYYZX"]

"""
StabTest should have Autweak=S3 and Autstrong=S2. Try using ListPerms function.
"""

def ListPerms(Gen):
    """
    Prints out permutations in weak and strong automorphism groups, as well as
    the sizes of the groups.
    Weak automorphism group is calculated first to improve efficiency.
    Input: List of strings with generators for stabilizing subgroup; e.g. Stab513
    """
    genList = GenToArray(Gen)
    C = Stabilizer(Gen)
    short_gen = [w[1:] for w in genList]
    short_code = [w[1:] for w in C]
    size = len(short_code[0])
    weakgroup = []
    stronggroup = []
    for p in Permutations(size):
        check = true
        for w in short_gen:
            if check == true:
                check = [w[p[i]-1] for i in range(size)] in short_code
            else:
                break
        if check == true:
            print(p)
            weakgroup.append(p)
    print("Size of weak automorphism group: " + str(len(weakgroup)))
    for p in weakgroup:
        check = true
        for w in genList:
            if check == true:
                check = ([w[0]]+[w[p[i]] for i in range(size)]) in C
            else:
                break
        if check == true:
            print(p)
            stronggroup.append(p)
    print("Size of strong automorphism group: " + str(len(stronggroup)))
\end{verbatim}

\end{appendices}

\end{document}